\documentclass[11pt]{scrartcl}%
\usepackage[utf8]{inputenc}
\usepackage{amsfonts}
\usepackage{amsmath}
\usepackage{amssymb}
\usepackage{graphicx}
\usepackage{fullpage}
\usepackage{hyperref}%
\usepackage{nicefrac}
\usepackage{enumerate}
\providecommand{\U}[1]{\protect\rule{.1in}{.1in}}
\newtheorem{theorem}{Theorem}

\newtheorem{corollary}[theorem]{Corollary}

\newtheorem{definition}[theorem]{Definition}

\newtheorem{lemma}[theorem]{Lemma}

\newtheorem{remark}[theorem]{Remark}

\newenvironment{proof}[1][Proof]{\noindent\textbf{#1.} }{\ \rule{0.5em}{0.5em}}

\def\be{\begin{equation}}
\def\ee{\end{equation}}
\def\bea{\begin{eqnarray}}
\def\eea{\end{eqnarray}}
\def\reff#1{(\ref{#1})}
\def\eps{\varepsilon}
\DeclareMathOperator{\cD}{\mathcal{D}}
\DeclareMathOperator{\cP}{\mathcal{P}}

\def\cH{{\cal H}}

\def\cE{{\cal E}}
\DeclareMathOperator{\tr}{Tr}  
\def\tD{{\widetilde{D}}}

\def\cDH{\cD(\cH)}
\def\cPH{\cP(\cH)}

\DeclareMathOperator{\supp}{supp}

\let\originalleft\left
\let\originalright\right
\renewcommand{\left}{\mathopen{}\mathclose\bgroup\originalleft}
\renewcommand{\right}{\aftergroup\egroup\originalright}

\usepackage{todonotes}

\begin{document}

\title{\textbf{A limit of the quantum R\'enyi divergence}}
\author{Nilanjana Datta and Felix Leditzky\\\textit{Statistical Laboratory, University of Cambridge,}\\\textit{Cambridge CB3 0WB, United Kingdom}}
\maketitle


\begin{abstract}
Recently, an interesting quantity called the quantum R\'enyi divergence (or ``sandwiched'' R\'enyi relative entropy) was defined for pairs of positive semi-definite operators $\rho$ and $\sigma$. It depends on a parameter $\alpha$ and acts as a parent quantity for other relative entropies which have important operational significances in quantum information theory: the quantum relative entropy and the min- and max-relative entropies. There is, however, another relative entropy, called the $0$-relative R\'enyi entropy, which plays a key role in the analysis of various quantum information-processing tasks in the one-shot setting. 
We prove that the $0$-relative R\'enyi entropy is obtainable from the quantum R\'enyi divergence only if $\rho$ and $\sigma$ have equal supports. 
This, along with existing results in the literature, suggests that it suffices to consider two essential parent quantities
from which operationally relevant entropic quantities can be derived -- the quantum R\'enyi divergence with parameter $\alpha \ge 1/2$, and the $\alpha$-relative R\'enyi entropy with $\alpha\in [0,1)$.

\end{abstract}

\section{Introduction}
A fundamental quantity in quantum mechanics which plays a particularly relevant role in quantum information theory is the $\alpha$-relative R\'enyi entropy (see e.g. \cite{Renyi} and \cite{milan1}), where the parameter $\alpha\in [0, \infty) \setminus \{1\}$.
Its applications range over problems as diverse as binary state discrimination, state compression, information transmission, catalytic transformations of states, as well as finding bounds on the communication complexity of certain distributed computation problems (see e.g.~\cite{LMW} and references therein).

Recently, a generalization of this quantity was proposed independently by  Wilde {\em{et al.}} \cite{WWY} and M\"uller-Lennert {\em{et al.}}\cite{ML}. The resulting quantity, which was referred to as {\em{``sandwiched'' R\'enyi relative entropy}} in the former paper, and {\em{quantum R\'enyi divergence}} in the latter\footnote{In this paper, we use the name {\em{quantum R\'enyi divergence}}}, is defined as follows (see also \cite{FL, Beigi}): For a density matrix $\rho$, a positive semi-definite operator $\sigma$ and a parameter $\alpha \in (0,1) \cup (1,\infty)$,
\be
 D_\alpha(\rho || \sigma):=
 \frac{1}{\alpha - 1} \log \left[ \tr\left(\sigma^{\frac{1-\alpha}{2\alpha}}\rho  \sigma^{\frac{1-\alpha}{2\alpha}}\right)^\alpha\right]
\ee
For $\alpha \ge 1$, we set $D_\alpha(\rho || \sigma)= \infty$ if $\supp\rho \not\subseteq \supp \sigma$.
Here $\supp\rho$ denotes the support of $\rho$, i.e., the span of eigenvectors of $\rho$ corresponding to non-zero eigenvalues. 
The above definition is easily extended to the case in which $\rho \ge 0$ but $\tr \rho \ne 1$, and our result (Theorem~\ref{thm:main-result} below) is valid for it too.

For certain ranges of values of the parameter $\alpha$, the quantum R\'enyi divergence, 
$D_\alpha(\rho || \sigma)$, has been proved to possess a host of interesting properties
desirable for a divergence measure 
 (see e.g. \cite{WWY, ML, FL, Beigi}). In particular, for $\alpha \ge 1/2$ it satisfies the data-processing inequality \cite{FL, Beigi}, i.e., it is 
monotonous under any completely positive trace-preserving (CPTP) map $\Lambda$:
\begin{align*}
 D_\alpha\left(\Lambda(\rho) ||\Lambda(\sigma)\right) \le  D_\alpha(\rho || \sigma) \quad {\hbox{for}} \,\,\alpha \ge 1/2.
\end{align*}
It satisfies joint convexity for $1/2\le \alpha \le 1$ \cite{FL}, and positivity if $\rho$ and $\sigma$ are both density matrices \cite{WWY, ML, Beigi}. 

Moreover, the quantum R\'enyi divergence acts as a parent quantity for other relative entropy quantites. If $[\rho, \sigma]=0$, it reduces to the standard $\alpha$-relative R\'enyi entropy: 
\begin{align*}
 \tD_\alpha(\rho || \sigma):= \frac{1}{\alpha -1} \log \left( \tr(\rho^\alpha \sigma^{1-\alpha})\right),
\end{align*}
and hence can be viewed as a generalization of  $\tD_\alpha(\rho || \sigma)$ to the non-commutative case. In the limit 
 $\alpha \to 1$, it reduces to the quantum relative entropy $$D(\rho||\sigma) := \tr(\rho \log \rho) - \tr( \rho \log \sigma),$$ whereas in the limit $\alpha \to \infty$
it reduces to the max-relative entropy \cite{datta08}: $$D_{\max}(\rho||\sigma) :=\inf \{\gamma: \rho \le 2^\gamma \sigma\}.$$ Moreover, for $\alpha=1/2$ it is equal to the min-relative entropy \cite{dupuis}: $$D_{\min}(\rho||\sigma) := - 2 \log ||\sqrt{\rho}\sqrt{\sigma}||_1.$$

The min- and max- relative entropies are of particular relevance in the fast-developing field of one-shot quantum information theory. This is because characterizations of operational quantities (e.g. capacities of channels, data compression limits,
entanglement cost, distillable entanglement etc.) in the one-shot setting can be obtained in terms of smoothed entropies \cite{Renner} derived from them (see e.g. \cite{one-shot} and references therein). There is, however, another 
relative entropy which is of particular relevance in the one-shot setting. This is the $0$-relative R\'enyi entropy
\begin{align}\label{rel0}
 \tD_0(\rho||\sigma) := \lim_{\alpha \to 0} \tD_\alpha(\rho || \sigma)
= - \log\left(  \tr \Pi_\rho \sigma\right),
\end{align}
where $\Pi_\rho$ denotes the projection onto the support of $\rho$ (and obviously satisfies
$0\le \Pi_\rho \le I$, $\tr(\Pi_\rho\rho) = 1$).
It has a simple operational interpretation in the problem of binary state discrimination; namely, for two density matrices $\rho$ and $\sigma$, $2^{-\tD_0(\rho||\sigma)}$ is equal to the minimum probability of type II error, under the condition that the type I error is strictly zero\footnote{For a binary POVM $\{E, I-E\}$ used for the state discrimination, we define the probabilities of 
type I and type II errors respectively as
$\alpha(\rho, \sigma) = \tr\left((I-E)\rho \right)$ and $\beta(\rho, \sigma) = \tr\left(E\sigma \right)$ respectively.}. If instead one allows the probability of type I error to be bounded above by $\eps >0$, then the minimum  probability of type II error is given by a similar expression, but with $\tD_0(\rho||\sigma)$ replaced by a smoothed version obtained by replacing $\Pi_\rho$ in \reff{rel0} by an operator $0\le A \le I$, $\tr(A\rho) \ge 1- \eps$, and maximizing over all such operators. This yields the quantity 
$$\max_{\substack{0\le A \le I\\\tr(A\rho) \ge 1- \eps}}\left( - \log \left( \tr\left(A \sigma\right)\right) \right).$$
The above quantity is usually referred to as
the {\em{hypothesis testing relative entropy}} \cite{WR, dupuis}, and denoted as $D_H^\eps(\rho||\sigma)$. It has various interesting properties and applications in one-shot quantum information theory, e.g. bounds on the one-shot capacity of a classical-quantum channel can be expressed in terms of it \cite{WR}; it also arises in the analysis of one-shot lossy data compression \cite{lossy}. An alternative smoothed version of $\tD_0(\rho||\sigma)$ occurs naturally in the characterization of the one-shot entanglement cost under local operations and classical communication (LOCC) \cite{BD}. In fact, the (unsmoothed) $0$-relative R\'enyi entropy provides an explicit characterization of the entanglement cost of a bipartite state in the case of {\em{perfect}} entanglement dilution in the one-shot setting \cite{BD}: the minimum number of Bell states needed to create a single copy of a bipartite state $\rho_{AB}$ {\em{perfectly}}, by two distant parties (Alice and Bob, say) using LOCC is given by the following expression:
$$ \min_{\cE}\max_{\sigma_R} \left\{ - \tD_0 \left( \rho_{AR}^\cE||  I_A \otimes \sigma_R\right) \right\},$$
where the maximization is over density matrices $\sigma_R$ on a Hilbert space $\cH_R$, and the minimization is over all possible pure-state ensembles $\cE=\{p_i, |\psi^i_{AB}\rangle\}$ such that $\rho_{AB} = \sum_i p_i |\psi^i_{AB}\rangle \langle \psi^i_{AB}|$, and  $\rho_{AR}^\cE = \tr_B  \rho_{ABR}^\cE$ where 
$$ \rho_{ABR}^\cE= \sum_i p_i |i_R\rangle \langle i_R|\otimes |\psi^i_{AB}\rangle \langle \psi^i_{AB}|.$$

Hence, it is evident that the $0$-relative R\'enyi entropy plays an important role in one-shot quantum information theory. Thus, it is interesting to ask whether the quantum R\'enyi divergence also acts as a parent quantity for it, in the general case in which $[\rho, \sigma]\ne 0$. In this paper, we answer this question in the affirmative when $\rho$ and $\sigma$ have equal supports. However, the $0$-relative R\'enyi entropy is not obtainable from the quantum R\'enyi divergence when the supports of $\rho$ and $\sigma$ are unequal. We prove the latter via an explicit counterexample. 

Our result is given by Theorem~\ref{thm:main-result} below, and its implications are discussed in the Conclusions.

\section{The result}
Throughout the paper $\cH$ denotes a finite-dimensional Hilbert space. We denote by $\cPH$ the set of positive semi-definite operators on $\cH$ and by $\cDH$ the set of density operators on $\cH$, i.e.~operators $\rho\in\cPH$ with $\tr\rho=1$. 

Our result is stated in the following theorem.
\begin{theorem}\label{thm:main-result}
For $\rho\in \cDH$ and $\sigma\in\cPH$ with ${\rm{supp}}\, \rho = {\rm{supp}}\, \sigma$, the following identity holds 
\be\label{eq:main}
\lim_{\alpha\rightarrow 0}D_\alpha(\rho||\sigma) = \tD_0(\rho||\sigma).
\ee
The above identity does not necessarily hold if ${\rm{supp}}\, \rho \subset {\rm{supp}}\, \sigma$.
\end{theorem}
\smallskip

The claim in \reff{eq:main} follows from Lemma \ref{lem:upper-bound} and Lemma \ref{lem:lower-bound} below. In  Lemma \ref{lem:upper-bound} we prove
that the left hand side of \reff{eq:main} is upper bounded by $\tD_0(\rho||\sigma)$, for any $\rho\in \cDH$ and $\sigma\in\cPH$ with ${\rm{supp}}\, \rho \subseteq {\rm{supp}}\, \sigma$. In Lemma \ref{lem:lower-bound} we
prove that, {\em{if}} the supports of $\rho$ and $\sigma$ are identical, then 
the left hand side of \reff{eq:main} is also lower bounded by $\tD_0(\rho||\sigma)$.

Finally we provide an explicit, analytical counterexample to \reff{eq:main}
for the case in which the supports of $\rho$ and $\sigma$ are not equal.

\subsection{Upper bound for $\lim_{\alpha\rightarrow 0} D_\alpha(\rho||\sigma)$}
Bounding the quantum R\'enyi divergence $D_\alpha(\rho||\sigma)$ from above by the $\alpha$-relative R\'enyi entropy $\tD_\alpha(\rho||\sigma)$ is achieved by employing the Araki-Lieb-Thirring inequality \cite{Ar90, LT}, which we state here:
\begin{theorem}
Let $A$ and $B$ be positive semi-definite operators on a finite-dimensional Hilbert space $\cH$ and $q\geq 0$. Then the following holds:
\begin{enumerate}[{\normalfont (i)}]\label{thm:araki}
\item $\tr\left(B^{\nicefrac{1}{2}}AB^{\nicefrac{1}{2}}\right)^{rq} \leq \tr\left(B^{\nicefrac{r}{2}}A^rB^{\nicefrac{r}{2}}\right)^q$ for $r\geq 1$
\item $\tr\left(B^{\nicefrac{1}{2}}AB^{\nicefrac{1}{2}}\right)^{rq} \geq \tr\left(B^{\nicefrac{r}{2}}A^rB^{\nicefrac{r}{2}}\right)^q$ for $0\leq r< 1$
\end{enumerate}
\end{theorem}
This allows us to prove the following
\begin{lemma}\label{lem:upper-bound}
For all $\alpha >0$, $\rho \in \cDH$, and $\sigma\in\cPH$,  with ${\rm{supp}}\, \rho \subseteq {\rm{supp}}\, \sigma$, we have\\ $D_\alpha(\rho||\sigma)\leq \tD_\alpha(\rho||\sigma)$.
\end{lemma}
\begin{proof}
Let $\beta:=\frac{1-\alpha}{2\alpha}$ and consider $\alpha\in (0,1)$. Setting $q=1$, $r=\alpha$, $A=\rho$ and $B^{\nicefrac{1}{2}}=\sigma^\beta$ in Theorem~\ref{thm:araki} (ii) yields
\begin{align}
\tr(\sigma^\beta\rho\sigma^\beta)^\alpha &\geq \tr(\sigma^{\frac{1-\alpha}{2}}\rho^\alpha\sigma^{\frac{1-\alpha}{2}}) = \tr(\rho^\alpha\sigma^{1-\alpha})\label{eqn:araki-estimate}
\end{align}
where the equality sign follows from the cyclicity of the trace. Since the logarithm is a monotonous function and ${\alpha-1}<0$ for $\alpha\in(0,1)$, we obtain from \eqref{eqn:araki-estimate} 
that
$$D_\alpha(\rho||\sigma)=\frac{1}{\alpha-1}\log\tr(\sigma^\beta\rho\sigma^\beta)^\alpha \leq \frac{1}{\alpha-1}\log\tr(\rho^\alpha\sigma^{1-\alpha}) = \tD_\alpha(\rho||\sigma),$$ proving the claim for $\alpha \in (0,1)$. The claim for
$\alpha > 1$ follows analogously from Theorem~\ref{thm:araki} (i), as was mentioned in \cite{WWY}.
\end{proof}
\begin{remark}
{\em{Note that Lemma \ref{lem:upper-bound} implies that $$\lim_{\alpha\rightarrow 0}D_\alpha(\rho||\sigma) \leq \lim_{\alpha\rightarrow 0}\tD_\alpha(\rho||\sigma) = \tD_0(\rho||\sigma).$$}}
\end{remark}

\subsection{Lower bound for $\lim_{\alpha\rightarrow 0}D_\alpha(\rho||\sigma)$ when ${\rm{supp}}\, \rho = {\rm{supp}}\, \sigma$}
To obtain a lower bound of $\lim_{\alpha\rightarrow 0} D_\alpha(\rho||\sigma)$, when ${\rm{supp}}\, \rho = {\rm{supp}}\, \sigma$, we use the concept of ``pinching'' defined as follows (see e.g. \cite[Appendix A]{OH01}):
\begin{definition}[Pinching]\label{def:pinching}
Given an operator $B \in \cPH$, let $B= \sum_{I=1}^{n(B)} b_iP_i$ be its spectral
decomposition, where $n(B)$ is the number of distinct eigenvalues of $B$, and $P_i$ is the projection 
corresponding to the eigenvalue $b_i$, $i=1,2,\ldots, n(B)$. The following map
\begin{align*}
\cE_B: \omega\in \cPH \mapsto \cE_B(\omega):= \sum_{i=1}^{n(B)}P_i \omega P_i \,\,\in \cPH,
\end{align*}
is called pinching. Note that the set of projectors $\{P_i\}_{i=1}^{n(B)}$ satisfy the completeness relation $\sum_{i=1}^{n(B)} P_i = I$,
and are said to constitute a projection-valued measure (PVM).
\end{definition}
With this definition, the following lemma holds:
\begin{lemma}\label{lem:pinching}
Given a PVM $\{P_i\}_{i=1}^{n(B)}$, for all $\rho \in \cD(\cH)$,
\begin{align*}
\rho \le n(B) \cE_B(\rho).
\end{align*}
\end{lemma}
\begin{proof}
See \cite[Appendix B, Lemma 2]{OH01}.
\end{proof}

We employ a slight adaption of the above lemma, given by the following corollary, which can easily be proved using Gram-Schmidt orthogonalization.
\begin{corollary}\label{cor:pinching}
Consider $\sigma \in \cPH$ with $n:= {\rm{dim}} \cH$ and its eigenvalue decomposition
\begin{align}
\sigma = \sum_{j=1}^n s_j \pi_j, \quad {\hbox{with}} \quad \pi_j = |\psi_j \rangle\langle \psi_j|,\label{eqn:sigma-eig-decomp}
\end{align}
with eigenvalues $s_j$ listed in decreasing order, repeated according to their multiplicities, and rank-one projectors $\pi_j = |\psi_j\rangle \langle \psi_j|$ taking care of degeneracies. 
Then for all $\rho \in \cDH$:
\begin{align*}
 \rho \le n \sum_{j=1}^n \pi_j \rho \pi_j.
\end{align*}
\end{corollary}

Using the above corollary, we are able to prove that, {\em{if}} $\supp\rho=\supp\sigma$ a lower bound for $\lim_{\alpha\rightarrow 0} D_\alpha(\rho||\sigma)$ is indeed given by $\tD_0(\rho||\sigma)$:
\begin{lemma}\label{lem:lower-bound}
For $\rho\in \cDH$ and $\sigma\in\cPH$, with $\supp\rho=\supp\sigma$ we have
$$\lim_{\alpha\rightarrow 0}D_\alpha(\rho||\sigma)\geq \tD_0(\rho||\sigma).$$
\end{lemma}
\begin{proof}
Let $\rho \in \cDH$ and $\sigma \in \cPH$ with $\supp\rho = \supp\sigma$. Let $n:= {\rm{dim}} \cH$ and $\beta:=\frac{1-\alpha}{2\alpha}$, with $\alpha \in (0,1)$. Furthermore, let $\sigma$ have the eigenvalue decomposition given by \eqref{eqn:sigma-eig-decomp}. Let the analogous eigenvalue decomposition for $\rho$ be $\rho = \sum_{i=1}^n r_i |\phi_i\rangle \langle \phi_i|$.
We also define
$$
\mu_j := \sum_{i=1}^n r_i |\langle \phi_i |\psi_j\rangle |^2.
$$

Then, by Corollary~\ref{cor:pinching},
\begin{align}
C_\alpha := \sigma^\beta \rho \sigma^\beta &\le n \sum_{j=1}^n \pi_j \left(\sigma^\beta \rho \sigma^\beta \right) \pi_j\nonumber\\
&= n \sum_{j=1}^{n} s_j^{2\beta} \pi_j \rho \pi_j\nonumber\\
&= n \sum_{j=1}^{n} \left(s_j^{2\beta} \sum_{i=1}^n r_i|\langle \phi_i |\psi_j\rangle |^2 \right)\pi_j\nonumber\\
&= n  \sum_{j=1}^n \left(s_j^{2\beta} \mu_j\right)\pi_j\nonumber\\
&= n Q_\alpha.\label{calpha}
\end{align}
where
$Q_\alpha :=  \sum_{j=1}^n \left(s_j^{2\beta} \mu_j\right)\pi_j$. 
From \reff{calpha} and the operator monotonicity of $x \mapsto x^\alpha$ for $0 < \alpha < 1$, it follows  that 
\begin{align*}
C_\alpha^\alpha \le n^\alpha Q_\alpha^\alpha,
\end{align*}
Since $Q_\alpha$ is already spectrally represented, we have that 
\begin{align}
\label{dalpha}
Q_\alpha^\alpha =  \sum_{j=1}^n s_j^{1-\alpha} \mu_j^\alpha \pi_j.
\end{align}
Hence,
$$
\tr C_\alpha^\alpha \le n^\alpha \tr Q_\alpha^\alpha = n^\alpha \sum_{j\in J} s_j^{1-\alpha} \mu_j^{\alpha}.
$$
Moreover, since $\alpha <1$,
\begin{align*}
\frac{1}{\alpha - 1} \log \tr C_\alpha^\alpha \ge \frac{\alpha}{\alpha - 1} \log n + \frac{1}{\alpha - 1} \log \bigg(  \sum_{j=1}^n s_j^{1-\alpha}  \mu_j^{\alpha} \bigg)
\end{align*}
and hence
\be\label{last1}
\lim_{\alpha \to 0} D_\alpha(\rho||\sigma) \ge  - \log \bigg( \sum_{j=1}^n s_j \bigg)= - \log \left(\tr \sigma\right).
\ee
Moreover, since $\supp \rho= \supp \sigma$, we have 
\be\label{last2}
 \tD_0(\rho||\sigma) = - \log \left( \tr( \Pi_\rho \sigma)\right) = - \log \left(\tr \sigma\right).
\ee
From \reff{last1} and \reff{last2} it follows that if ${\rm{supp}} \, \rho =
{\rm{supp}} \, \sigma$,  
$$
\lim_{\alpha \to 0} D_\alpha(\rho||\sigma) \ge  \tD_0(\rho||\sigma).
$$
\end{proof}
\subsection{A counterexample to eq.\reff{eq:main} when $\rho$ and $\sigma$ have unequal supports}
Let 
$$
\rho=\begin{pmatrix}1&0\cr0&0\end{pmatrix}\ ,\quad \sigma=\begin{pmatrix}1& c\cr c&1\end{pmatrix}\ , \quad 0 < c < 1.
$$

Note that $\sigma$ is positive-definite, $\supp \rho \subset \supp \sigma$
and $[\rho, \sigma] \ne 0$. Moreover,
$\rho = \Pi_\rho$ and hence,
\be\label{eqq1}
\tD_0(\rho||\sigma) = - \log \left( \tr( \rho \sigma)\right) = - \log 1 = 0.
\ee
Let us set $\beta:= \frac{1-\alpha}{2\alpha}$ with $\alpha \in (0,1)$. Then the eigenvalues of $C_\alpha:= \sigma^\beta \rho \sigma^\beta$ are given by
$$ \lambda_1 =  \frac{1}{2} \left((1+c)^{2\beta} + (1-c)^{2\beta}\right)\quad {\hbox{and}} \quad \lambda_2 =0.$$
Using the fact that $\lim_{\alpha \to 0} f(\alpha)^{g(\alpha)} = \exp\left(\lim_{\alpha \to 0} \frac{\ln f(\alpha)}{\left(g(\alpha)\right)^{-1}}\right)$, 
one obtains the following:
\be\label{eqq2}
\lim_{\alpha \to 0} D_\alpha(\rho||\sigma)= \lim_{\alpha \to 0} \frac{1}{\alpha - 1} \log \left(\tr C_\alpha^\alpha\right) = - \log (1 + c) <0.
\ee
From \reff{eqq1} and \reff{eqq2} it follows that in this case,
$$\lim_{\alpha \to 0} D_\alpha(\rho||\sigma) \ne \tD_0(\rho||\sigma).$$

\section{Conclusions}
We prove that the $0$-relative R\'enyi entropy  $\tD_0(\rho||\sigma)$ is obtainable from the quantum R\'enyi divergence $\tD_\alpha(\rho||\sigma)$ 
if $\rho$ and $\sigma$ have equal supports. 
While $\tD_\alpha(\rho||\sigma)$ serves as an upper bound for $D_\alpha(\rho||\sigma)$ for arbitrary $\alpha\in [0,1)$ as long as $\supp\rho\subseteq \supp\sigma$, the condition of equal supports is essential for proving that  $\tD_0(\rho||\sigma)$ also serves as a lower bound for $\lim_{\alpha\rightarrow 0}D_\alpha(\rho||\sigma)$, and hence in establishing equality. This is emphasized by the counterexample given in the last subsection.

Our result implies that for values of the parameter $\alpha$ in the range $[0,1/2)$, the relative entropy which plays a fundamental role is the $\alpha$-relative R\'enyi entropy, and not the quantum R\'enyi divergence. This claim is further strengthened by the fact that the quantum R\'enyi divergence does not satisfy the data-processing inequality (a property desirable for a relative entropy) in this parameter range. The latter follows from numerical counterexamples, reported in \cite{ML}, and also found independently by us. This observation, along with existing results in the literature, suggests that, as
far as operationally relevant entropic quantities are concerned, the right parent quantity depends on the value of $\alpha$.
 For example, there are various operational interpretations of the R\'enyi entropy of order $\alpha$, 
for all $\alpha \ge 0$ (see e.g.~\cite{renyi2} for a few of these). This is given by $S_\alpha(\rho)= - \tD_\alpha(\rho||\sigma) \equiv -D_\alpha(\rho||\sigma)$. For a bipartite state, the conditional R\'enyi entropy of order $\alpha$ arises in entanglement theory (see e.g.\cite{rencond}). However, this quantity is defined as the difference of two R\'enyi entropies of order $\alpha$, and so can equivalently
be obtained from $D_\alpha$ or $\tD_\alpha$. However, the quantum Chernoff bound, which is an important quatity arising in binary state discrimination,
is expressible in terms of the $\alpha$-relative R\'enyi entropy with $\alpha \in [0,1)$ but not in terms of the quantum R\'enyi divergence. (See also \cite{milan1} and \cite{hayashi1}).
On the other hand, the min- and max- relative entropies, which play pivotal roles in one-shot information theory are obtainable from the quantum R\'enyi divergence (for $\alpha = 1/2$ and $\alpha \to \infty$, respectively) but not from the  $\alpha$-relative R\'enyi entropy. The above observations suggest that it suffices to consider two essential parent quantities
from which operationally relevant entropic quantities can be derived. These are the quantum R\'enyi divergence with parameter $\alpha \ge 1/2$, and the $\alpha$-relative R\'enyi entropy with $\alpha\in [0,1)$.

\subsection*{Acknowledgements} We are grateful to Yuri Suhov for interesting discussions. We would like to thank Serge Fehr, Marco Tomamichel and Mark Wilde for helpful comments.

\end{document}